\documentclass[12pt]{article}

\usepackage{amsfonts,amscd,amssymb}
\usepackage{amsthm,amsmath,natbib}
\usepackage{bbm} 
\usepackage{times,graphicx}
\usepackage{mathrsfs}
\usepackage{enumerate}
\usepackage{subfigure}
\usepackage{appendix}
\usepackage{pdfpages}
\usepackage[margin=1in]{geometry}
\usepackage{setspace}
\usepackage{tikz}
\usepackage{todonotes}

\RequirePackage[colorlinks,citecolor=blue,urlcolor=blue]{hyperref}

\usepackage{natbib}

\newtheorem{thm}{Theorem}
\newtheorem{lemma}{Lemma}

\newtheorem{remark}{Remark}

\newtheorem{mydef}{Definition}

\newcommand{\mb}{\mathbf}

\newcommand{\mc}{\mathcal}

\newcommand{\one}{\mathbf{1}}

\newcommand{\tree}{\mathbb{T}}

\newcommand{\var}{\mathrm{Var}}

\newcommand{\cov}{\mathrm{Cov}}

\newcommand{\transpose}{^{\tiny \top}}
\newcommand{\V}{\mathbf{V}}
\newcommand{\KL}{\text{KL}}
\newcommand{\TV}{\text{TV}}
\newcommand{\sign}{\text{sign}}

\title{When can we reconstruct the ancestral state? A unified theory}
\date{}
\author{
Lam Si Tung Ho \\
Department of Mathematics and Statistics \\
Dalhousie University, Halifax, Nova Scotia, Canada
\and
Vu Dinh \\
Department of Mathematical Sciences \\
University of Delaware
}

\begin{document}
\maketitle

\begin{abstract}
Ancestral state reconstruction is one of the most important tasks in evolutionary biology.
Conditions under which we can reliably reconstruct the ancestral state have been studied for both discrete and continuous traits.
However, the connection between these results is unclear, and it seems that each model needs different conditions.
In this work, we provide a unifying theory on the consistency of ancestral state reconstruction for various types of trait evolution models.
Notably, we show that for a sequence of nested trees with bounded heights, the necessary and sufficient conditions for the existence of a consistent ancestral state reconstruction method under discrete models, the Brownian motion model, and the threshold model are equivalent.
When tree heights are unbounded, we provide a simple counter-example to show that this equivalence is no longer valid.
\end{abstract}

\clearpage

\section{Introduction}

The evolution of biological features, such as genotypes and phenotypes, is often assumed to follow Markov processes along a phylogenetic tree \citep{felsenstein2004inferring}. Under these models, each internal node in the tree depicts a speciation event when an ancestral lineage splits into two new ones. 
The descendant lineages inherit the ancestral state of their most recent common ancestor and then evolve independently from each other \citep{steel2016phylogeny}. 
One important task in evolutionary biology is reconstructing the ancestral state from observations at the leaves of a given tree. 
This problem, usually referred to as \emph{ancestral state reconstruction} or \emph{root reconstruction}, helps answer many questions about the underlying evolutionary process and directly affects the efficiency and accuracy of other phylogenetic estimates \citep{maddison1994phylogenetic, liberles2007ancestral, thornton2004resurrecting, ho2021ancestral}.
One important application is to infer the origin of epidemics \citep{faria2014early, gill2017relaxed}.

In recent years, evolutionary data for a wide variety of species are increasingly available, and the problem of ancestral reconstruction based on hundreds or thousands of leaves is becoming commonplace. 
It is well-known that sampled data at the leaves of the tree cannot be considered independent since closely related species are expected to have similar characteristics \citep{felsenstein1985phylogenies}. 
Previous works in the field indicate that in this setting, basic statistical properties should not be taken for granted \citep{ane2008analysis, li2008more, ho2013asymptotic, ho2014intrinsic, ane2017phase, ho2021ancestral}. 
For example, one of the most desired properties of good estimation methods is consistency (which dictates that the estimator converges to the true value as the number of leaves increases), but even rigorous methods such as Maximum likelihood estimator (MLE) could be inconsistent in phylogenetic settings. 
Characterizing conditions under which the ancestral state can be reliably estimated has become an active research direction.

Perhaps, \citet{ane2008analysis} provides the most notable result for reconstructing the ancestral state of continuous traits.
In this paper, the author derives a necessary and sufficient condition for the consistency of the MLE under the Brownian motion (BM) model.
The condition involves the covariance matrix $\V_n$ whose components are the times of shared ancestry between leaves, that is, the element in $i$th-row and $j$-th column, $V_{ij}$, is the length shared by the paths from the root to the leaves $i$ and $j$.
Specifically, the MLE is consistent if and only if $\one\transpose \V_n^{-1} \one \to \infty$.
For discrete models, \cite{fan2018necessary} show that under a certain root density assumption, referred to as the big bang condition, it is possible to identify a subset of leaves that are sufficiently independent. This enables the derivations of a necessary and sufficient condition for the existence of a consistent ancestral state reconstruction method under discrete models.

Despite the usefulness of these results, the connection between them is unclear. In principle, they consider different stochastic processes, study distinct aspects of the problem and focus on conditions with seemingly unrelated mathematical formulations. For example, \cite{ane2008analysis} specifically studies the MLE, while \cite{fan2018necessary} consider a more abstract question of the existence of a consistent estimator. 
It seems that the consistency property for each trait evolution model needs different conditions.
In this work, we attempt to bridge this gap by showing that when the sequences of trees are nested and have bounded heights, which corresponds to the natural setting where data from new species are continually being collected, the two geometric conditions of \cite{ane2008analysis} and \cite{fan2018necessary} are equivalent. 
As a consequence, they are the necessary and sufficient condition for the existence of a consistent ancestral state reconstruction method under the BM model and discrete models.
We also show that the results extend to the threshold model \citep{felsenstein2012comparative,revell2014ancestral}, thus providing a unifying perspective on the consistency of ancestral state reconstruction procedures across a wide range of popular phylogenetic Markov processes.
Finally, we give a simple counter-example to show that when tree heights are unbounded, these conditions are not equivalent and neither of them is a sufficient condition for the existence of a consistent ancestral state reconstruction method under discrete models.

\section{Settings}

We consider a sequence of nested trees $\tree_n$, meaning $\tree_{n-1}$ is a subtree of $\tree_n$ for all $n$.
It is worth noticing that this setting represents the situation when more species are continually sampled and added to the data set.
This is a common setup for theoretical studies of trait evolution models \citep{fan2018necessary, ho2013asymptotic}.
We denote the observed trait values at the leaves of $\tree_n$ by $\mb Y_n = (Y_k)_{k=1}^n$.
Without loss of generality, we assume that $\tree_n$ has $n$ species and the root of all trees is the same species.
Furthermore, we assume that distances from this root to the leaves are uniformly bounded by $H$.
The goal of ancestral state reconstruction is to estimate the trait value of this root from the trait values at the leaves.

In this paper, we study three different types of trait evolution models: BM, discrete, and threshold models.
As we already discussed, all three (types of) models follow Markovian dynamics along phylogenetic trees where at each internal node, descendant lineages inherit the value from the parent lineage just prior to the speciation event. Conditional of their starting value, each lineage then evolves independently of the sister lineages. 

\paragraph{Brownian motion model}
The BM model assumes that a continuous phenotype evolves along a tree according to a Brownian motion.
Under the BM model, the observations $\mb Y_n = (Y_k)_{k=1}^n$ follow a Gaussian distribution $\mc{N}(\mu, \sigma^2 \V_n)$. 
Here, $\mu$ is the ancestral state, $\sigma^2$ is the variance of the BM, and $\V_n = (t_{ij})$ depends on the tree where $t_{ij}$ is the distance from the root to the most recent common ancestor of leaves $i$ and $j$ \citep{ane2008analysis}.
We visualize the evolution of a trait along a tree under the BM model in Figure \ref{fig:BM}.

\begin{figure}[h]
\centering
\includegraphics[width = 0.8\textwidth]{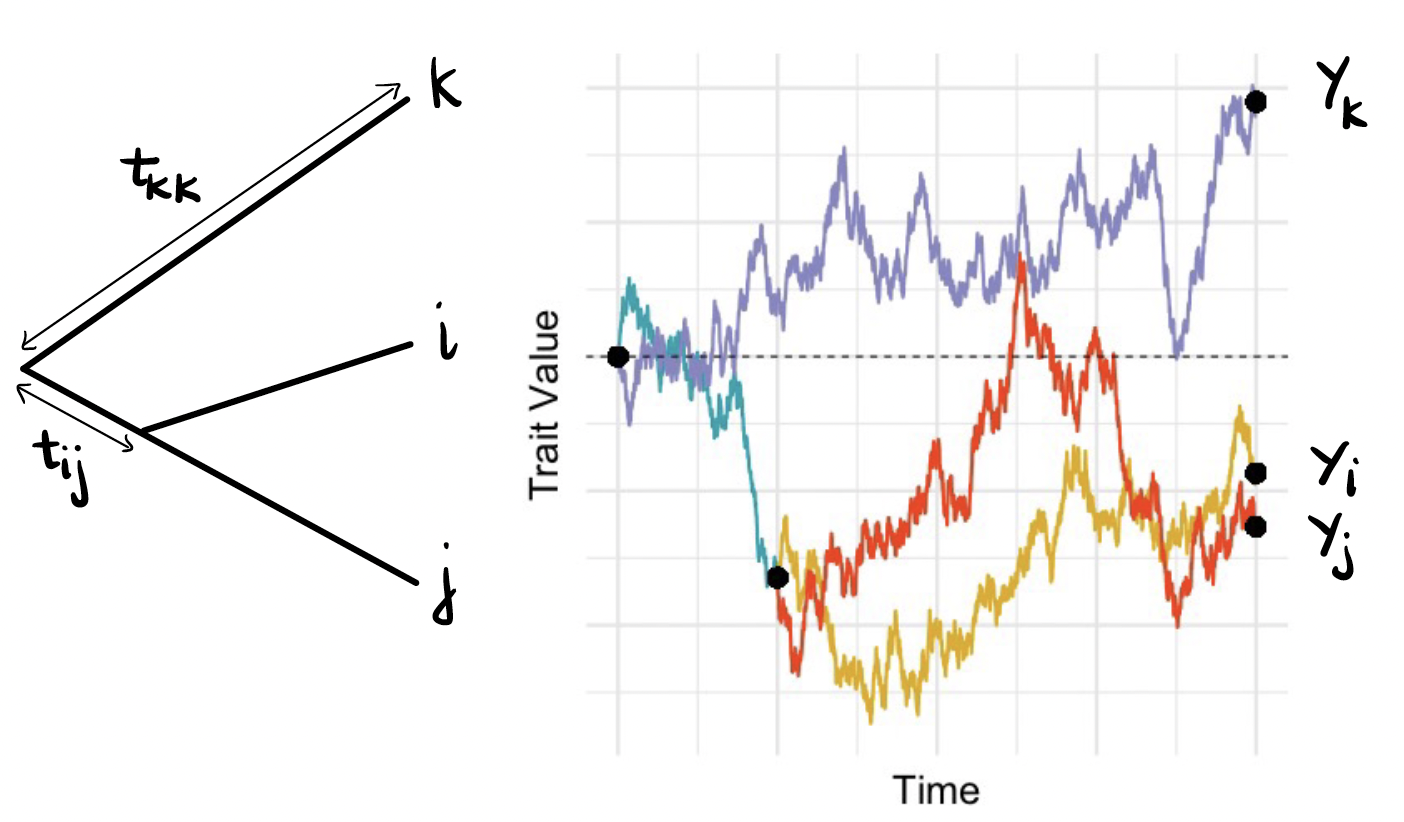}
\caption{Visualization of a BM process on a tree (right). The distance from the root to the most recent common ancestor of leaves $i$ and $j$ is $t_{ij}$, and the distance from the root to leaf $k$ is $t_{kk}$.}
\label{fig:BM}
\end{figure}

Maximum likelihood estimator (MLE) is the most popular method for reconstructing the ancestral state.
Under the BM model, the MLE has an analytic formula
\[
\hat \mu_n = (\one\transpose \V_n^{-1} \one)^{-1}(\one\transpose \V_n^{-1} \mb{Y}_n).
\]
\citet{ane2008analysis} provides the following necessary and sufficient condition for the consistency of the MLE:

\begin{lemma}[\citet{ane2008analysis}]
Under the BM model, the MLE of the ancestral state is consistent if and only if $\one\transpose\V_n^{-1}\one \to + \infty$.
\label{lem:Ane}
\end{lemma}

\paragraph{Discrete models}
These models assume that traits evolve along the tree according to a continuous-time Markov chain with finite state-space.
\citet{fan2018necessary} focus on models that satisfy the ``initial-state identifiability": all rows of the transition probability matrix $\mb{P}_t$ of the Markov chain are distinct for all $t$.
Throughout the paper, we also require that for two states $i, j$ (do not necessarily distinct), $P_{ij}(t) > 0$ for some $t > 0$.
We refer to models satisfying those two conditions as \emph{regular discrete models}.
It is worth noticing that popular evolution models, such as two-state, Jukes-Cantor, and GTR \citep{felsenstein2004inferring} (with positive transition rates) are regular discrete models.
\citet{fan2018necessary} derive a necessary and sufficient condition for the existence of a consistent estimator for the ancestral state, called the big bang condition.
To understand the big bang condition, let us first introduce some notations.
For a tree $\tree$, a truncated tree at level $s$ of $\tree$, denoted by $\tree(s)$, is the tree obtained by truncating $\tree$ at distance $s$ from the root.
We denote the set of leaves of a tree $\tree$ by $\partial \tree$ and denote the cardinality of a set $A$ by $| A |$.

\begin{mydef}[Big bang condition]
A nested sequence of trees $(\tree_n)_{n=1}^\infty$ satisfies the big bang condition if for all $s > 0$, we have $| \partial \tree_n(s)| \to \infty$ as $n \to \infty$.
\end{mydef}

\begin{lemma}[\citet{fan2018necessary}]
Under regular discrete models, there exists a consistent estimator for the ancestral state if and only if the big bang condition holds.
\label{lem:Roch}
\end{lemma}

We note that the ``downstream disjointness" condition in \citet{fan2018necessary} is not satisfied for regular discrete models and can be removed.

\paragraph{Threshold model}

Threshold model \citep{felsenstein2012comparative,revell2014ancestral} assumes that a binary phenotype ($\pm 1$) is driven by an underlying process that evolves along a tree according to a BM.
Let $\mb{Z}_n$ be the underlying process and $\mb{Y}_n$ be the observations at the leaves of the tree. 
Under threshold model, $\mb Z_n \sim \mc{N}(\mu, \sigma^2 \V_n)$ and 
\[
Y_i = 
\begin{cases}
1 & Z_i \geq 0 \\
-1 & Z_i < 0.
\end{cases}
\]
Figure \ref{fig:thres} visualizes the evolution of the underlying BM process and the corresponding observations.

We want to estimate the ancestral state at the root $\rho = \text{sign}(\mu)$.
To the best of our knowledge, there are no theoretical results for the problem of reconstructing the ancestral state under this threshold model.

\begin{figure}[h]
\centering
\includegraphics[width = 0.6\textwidth]{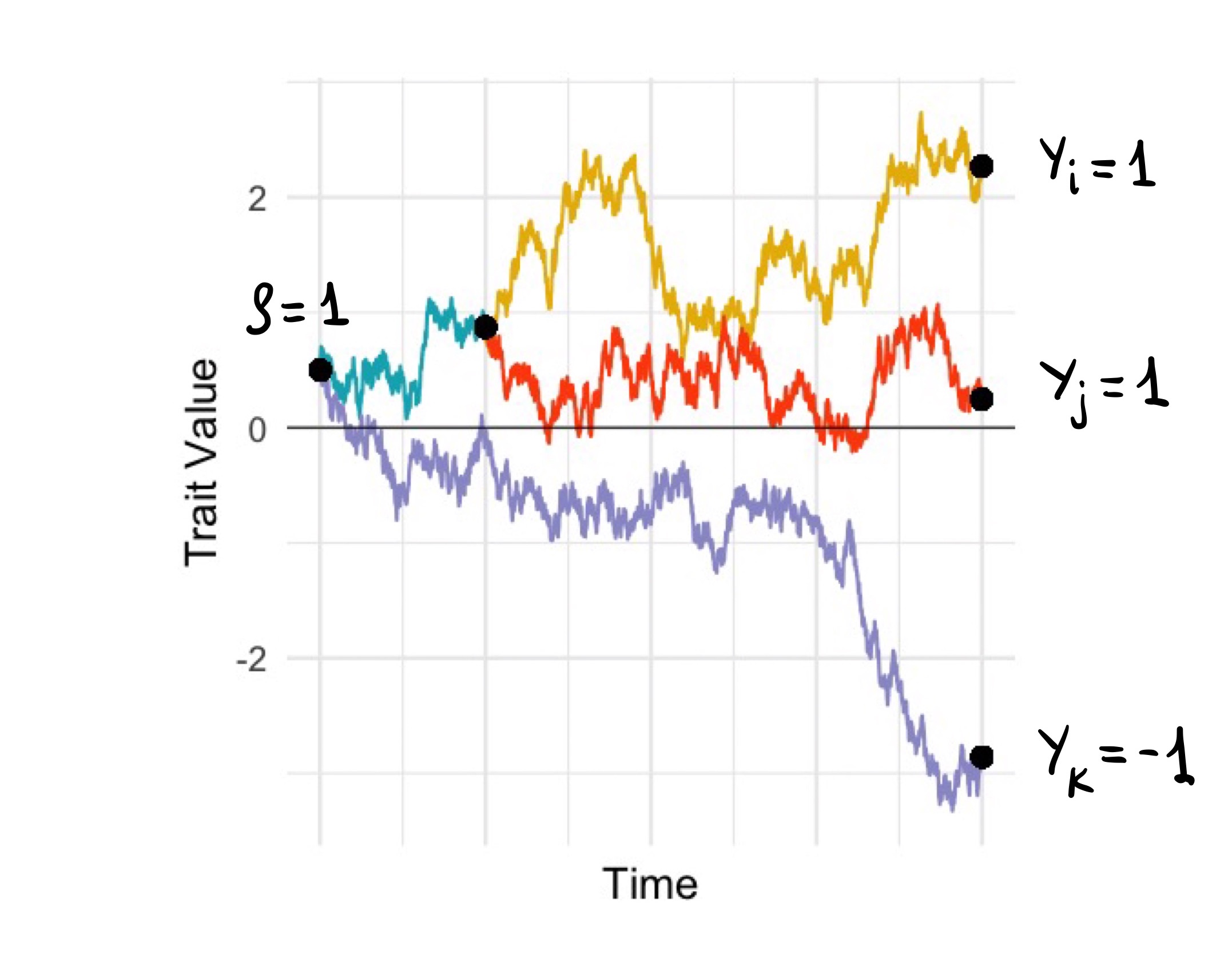}
\caption{Visualization of an underlying BM process and corresponding observations under the threshold model.}
\label{fig:thres}
\end{figure}

\section{Necessary and sufficient condition for consistency of ancestral state reconstruction}

While the results of Lemmas \ref{lem:Ane} and \ref{lem:Roch} are very useful, the connection between them is unclear. The derived conditions of the two models focus on seemingly unrelated mathematical formulations, and it seems that the consistency property for each model needs to be studied separately.
In this work, we aim to bridge this gap by showing that in our setting, the two geometric conditions for discrete and continuous models are equivalent. We then extend the results to threshold models to showcase the generalizability of the result across a wide range of popular phylogenetic Markov processes.

\begin{thm}
Under our settings (a sequence of nested trees with bounded heights),
\begin{itemize}
\item The big bang condition is equivalent with the condition $\one\transpose\V_n^{-1}\one \to + \infty$.
\item These conditions are the necessary and sufficient condition for the existence of a consistent estimator for the ancestral state under the BM, regular discrete, and threshold models.
\end{itemize}
\label{thm:main}
\end{thm}

The flow of our proofs is as follows.
First, we prove that the condition $\one\transpose\V_n^{-1}\one \to + \infty$ is a necessary condition for the existence of a consistent estimator for the ancestral state under the BM model (Theorem \ref{thm:BMnecessary}).
Together with Lemma \ref{lem:Ane}, we conclude that $\one\transpose\V_n^{-1}\one \to + \infty$ is also the necessary and sufficient condition.
Next, we show that the big bang condition is also the necessary and sufficient condition for the existence of a consistent estimator for the ancestral state under the BM model (Theorem \ref{thm:BMbigbang}).
Therefore, big bang condition and the condition $\one\transpose\V_n^{-1}\one \to + \infty$ are equivalent.
Finally, we prove that under the threshold model, the condition $\one\transpose \V_n^{-1} \one \to + \infty$ is a necessary condition (Theorem \ref{thm:thresholdneccesary}) and the big bang condition is the sufficient condition (Theorem \ref{thm:thresholdsufficient}) for the existence of a consistent estimator for the ancestral state .

\subsection{Equivalence of consistency condition for discrete and continuous models}

\begin{thm}
Under the BM model, a necessary condition for the existence of a consistent estimator for the ancestral state is $\one\transpose\V_n^{-1}\one \to + \infty$.
\label{thm:BMnecessary}
\end{thm}

\begin{proof}
We only need to prove that if there exists a constant $C > 0$ such that $\one\transpose\V_n^{-1}\one \le C$ for all $n$, then there is no consistent estimator for the ancestral state.
Let $P_{\mu, \sigma^2}$ be the joint distribution of the observations $\mb{Y}_n$ under the BM model with mean $\mu$ and variance $\sigma^2$.
We have 
\[
\KL(P_{\mu_1,\sigma^2}, P_{\mu_2,\sigma^2}) = \frac{1}{2 \sigma^2} \one\transpose \V_n^{-1} \one(\mu_1 - \mu_2)^2 \leq \frac{C}{2 \sigma^2} (\mu_1 - \mu_2)^2.
\]
Here, $\KL(P_{\mu_1,\sigma^2}, P_{\mu_2,\sigma^2})$ denotes the Kullback-Leibler divergence from $P_{\mu_2,\sigma^2}$ to $P_{\mu_1,\sigma^2}$.
Let $d_{\TV}(P_{\mu_1,\sigma^2}, P_{\mu_2,\sigma^2})$ be the total variation distance between $P_{\mu_1,\sigma^2}$ and $P_{\mu_2,\sigma^2}$.
That is, $d_{\TV}(P_{\mu_1,\sigma^2}, P_{\mu_2,\sigma^2}) = \sup_{\mc{A}} | P_{\mu_1,\sigma^2}(\mc{A}) - P_{\mu_2,\sigma^2}(\mc{A}) |$.
Applying Vajda's inequality \citep{vajda1970note}, we have
\[
\KL(P_{\mu_1,\sigma^2}, P_{\mu_2,\sigma^2}) \geq \log \left ( \frac{1 + d_{\TV}(P_{\mu_1,\sigma^2}, P_{\mu_2,\sigma^2})}{1 - d_{\TV}(P_{\mu_1,\sigma^2}, P_{\mu_2,\sigma^2})} \right ) - \frac{2 d_{\TV}(P_{\mu_1,\sigma^2}, P_{\mu_2,\sigma^2})}{d_{\TV}(P_{\mu_1,\sigma^2}, P_{\mu_2,\sigma^2}) + 1}.
\]
Hence, $d_{\TV}(P_{\mu_1,\sigma^2}, P_{\mu_2,\sigma^2}) \leq d_0 < 1$.
Assume that there exists a consistent estimator $\hat \mu_n$ for the ancestral state. 
Define $\mc{A} = \{ \hat \mu_n \to \mu_1 \}$, we have
\[
P_{\mu_1,1}(\mc{A}) \to 1 \quad \text{and} \quad P_{\mu_2,1}(\mc{A}) \to 0,
\]
where $\mu_1 \ne \mu_2$.
Therefore,
\[
d_{\text{TV}}(P_{\mu_1,\sigma^2}, P_{\mu_2,\sigma^2}) \geq | P_{\mu_1,1}(\mc{A}) - P_{\mu_2,1}(\mc{A})| \to 1.
\]
This is a contradiction.
Therefore, there is no consistent estimator for the ancestral state. 
\end{proof}

\begin{thm}
Under the BM model, there exists a consistent estimator for the ancestral state if and only if the big bang condition holds.
\label{thm:BMbigbang}
\end{thm}

\begin{proof}
First, we will prove that if the big bang condition does not hold, then there exists a constant $C > 0$ such that $\one\transpose\V_n^{-1}\one \leq C$.
When the big bang condition does not hold, there exists $s >0$, $K >0 $ and $N > 0$ such that $| \partial \tree_n(s) | = K$ for all $n \geq N$.
Let $\epsilon$ be the smallest distance from the root to the internal nodes and leaves of $\tree_N(s)$.
We note that, by this construction, $\tree_n(\epsilon)$ is a fixed ultrametric star tree with height equals to $\epsilon$ for all $n \geq N$.
Let $I_1, I_2, \ldots, I_\ell$ be the leaves of $\tree_n(\epsilon)$ and $\tree_1, \tree_2, \ldots, \tree_\ell$ be the subtree of $\tree$ stemming from $I_1, I_2, \ldots, I_\ell$.
Then,
\[
\V_n = 
\begin{pmatrix}
\V_{\tree_1} + \epsilon \one \one\transpose & 0 & \cdots & 0\\
0 & \V_{\tree_2} + \epsilon \one \one\transpose & \cdots& 0 \\
\vdots & \vdots &\ddots & \vdots \\
0 & 0 & \cdots & \V_{\tree_\ell} + \epsilon \one \one\transpose
\end{pmatrix}
\]
By the Woodbury matrix identity, we have
\begin{align*}
\one\transpose\V_n^{-1}\one &= \sum_{i=1}^\ell {\one\transpose (\V_{\tree_i} + \epsilon \one \one\transpose)^{-1} \one} \\
& = \sum_{i=1}^\ell{ \frac{1}{(\one\transpose\V_{\tree_i}^{-1}\one)^{-1} + \epsilon}} \leq \frac{\ell}{\epsilon}.
\end{align*}

Next, we will prove that if the big bang condition holds, then there exists a consistent estimator for the ancestral state.
By the big bang condition, for any positive integer $m$, there exists $k_m > k_{m-1}$ such that $| \partial\tree_{k_m}(1/m) | \geq m$ with a convention that $k_0 = 0$. 
Thus, there exists a subtree of $m$ leaves of $\tree_{k_m}$ such that distances from the root to all internal nodes are less than $1/m$. Let $Y_{1,m}, \ldots, Y_{m,m}$ be the leaves of this subtree.
We define our estimator as follows:
\[
\hat \mu_n =  \frac{Y_{1,m} + Y_{2, m} + \cdots + Y_{m,m}}{m}, \quad k_m \leq n <  k_{m+1}.
\]
Note that $E(\hat \mu_n) = \mu$ and $\cov(Y_{i,m}, Y_{j,m}) = \sigma^2 t_{ij,m} \leq \sigma^2/m$ where $t_{ij,m}$ is the distance from the root to the most recent common ancestor of the leaves $Y_{i,m}$ and $Y_{j,m}$.
Therefore,
\begin{align*}
\var(\hat \mu_n) &= \frac{1}{m^2} \left (\sum_{i=1}^m{\var(Y_{i,m})} + 2\sum_{1 \leq i < j \leq m}{\cov(Y_{i,m}, Y_{j,m})} \right ) \\
&\leq \frac{1}{m^2} \left ( mH \sigma^2 + m(m-1)\frac{\sigma^2}{m} \right ) \leq \frac{H+1}{m}\sigma^2 \to 0.
\end{align*}
By Chebyshev's inequality, for all $\epsilon > 0$, we have
\[
P(|\hat \mu_n - \mu| \geq \epsilon) \leq \frac{\var(\hat \mu_n)}{\epsilon^2} \to 0.
\]
Hence, $\hat \mu_n$ is a consistent estimator.

\end{proof}

\begin{remark}
We note that the first part of the proof of Theorem \ref{thm:BMbigbang} also shows that the condition $\one\transpose\V_n^{-1}\one \to + \infty$ implies the big bang condition even without the assumption of bounded tree heights.
\end{remark}

\subsection{Necessary and sufficient condition for consistency of ancestral state reconstruction for threshold models}

\begin{thm}
Assume that $\one\transpose \V_n^{-1} \one$ are bounded. Then, there is no consistent estimator for the ancestral state under the threshold model.
\label{thm:thresholdneccesary}
\end{thm}

\begin{proof}
Let $P_{\mu, \sigma^2}$ and $Q_{\mu, \sigma^2}$ be the joint distribution of $\mb{Z}$ and $\mb{Y}$ respectively.
We have 
\[
\KL(Q_{\mu_1,\sigma^2}, Q_{\mu_2,\sigma^2}) \leq \KL(P_{\mu_1,\sigma^2}, P_{\mu_2,\sigma^2}) = \frac{1}{2 \sigma^2} \mb{1}^\top \mb{V}_n^{-1} \mb{1} (\mu_1 - \mu_2)^2.
\]
Applying Vajda's inequality \citep{vajda1970note}, we have
\[
\KL(Q_{\mu_1,\sigma^2}, Q_{\mu_2,\sigma^2}) \geq \log \left ( \frac{1 + d_{\TV}(Q_{\mu_1,\sigma^2}, Q_{\mu_2,\sigma^2})}{1 - d_{\TV}(Q_{\mu_1,\sigma^2}, Q_{\mu_2,\sigma^2})} \right ) - \frac{2 d_{\TV}(Q_{\mu_1,\sigma^2}, Q_{\mu_2,\sigma^2})}{d_{\TV}(Q_{\mu_1,\sigma^2}, Q_{\mu_2,\sigma^2}) + 1}.
\]
Hence,
\[
d_{\text{TV}}(Q_{\mu_1,\sigma^2}, Q_{\mu_2,\sigma^2}) \leq c < 1.
\]
Assume that there exists a consistent estimator $\hat \rho_n$ for the ancestral state $\rho = \sign(\mu)$. 
Define $\mc{A} = \{ \hat \rho_n = 1 \}$, we have
\[
Q_{1,1}(\mc{A}) \to 1 \quad \text{and} \quad Q_{-1,1}(\mc{A}) \to 0,
\]
which implies
\[
d_{\text{TV}}(Q_{1,1}, Q_{-1,1}) = \sup_{A}(|Q_{1,1}(A) - Q_{-1,1}(A)|) \geq |Q_{1,1}(\mc{A}) - Q_{-1,1}(\mc{A})| \to 1.
\]
This is a contradiction.
Therefore, there is no consistent estimator for $\rho = \sign(\mu)$. 
\end{proof}

\begin{lemma}[\citet{lancaster1957some}]
Let $(X,Y)$ be a bivariate normal distribution and two functions $f, g$ such that $E(f(X)^2) < + \infty$ and $E(g(Y)^2) < + \infty$.
Then
\[
\frac{|\cov(f(X), g(Y))|}{\sqrt{\var(f(X)) \var(g(Y))}} \leq \frac{|\cov(X,Y)|}{\sqrt{\var(X) \var(Y)}}.
\]
\label{lem:boundcorr}
\end{lemma}

\begin{thm}
Assume that big bang condition is satisfied. Then, there is a consistent ancestral state reconstruction method under the threshold model.
\label{thm:thresholdsufficient}
\end{thm}

\begin{proof}
Let $k_m$ be the increasing sequence constructed in the proof of Theorem \ref{thm:BMbigbang}.
There exist a subtree of $m$ leaves of $\tree_{k_m}$ such that distances from the root to all internal nodes are less than $1/m$. Let $Y_{1,m}, \ldots, Y_{m,m}$ be the leaves of this subtree.
Let $t_{i,m}$ be the distance from the root to the leaf $Y_{i,m}$, and $t_{ij,m}$ be the distance from the root to the most recent common ancestor of the leaves $Y_{i,m}$ and $Y_{j,m}$.

If there is a sequence of leaves whose distance to the root converges to $0$, then their trait values form a trivial consistent estimator of the ancestral state.
Formally, denote $\tau_m = \min_i t_{i,m}$ and $s_m = \arg\min_i t_{i,m}$.
If there exists a subsequence $\tau_{m_u} \to 0$, then $Y_{s_{m_u}, m_u}$ is a trivial consistent estimator for the ancestral state $\rho = \sign(\mu)$.

On the other hand, if there exists $\alpha > 0$ such that $\tau_m \geq \alpha$ for all $m$, we will prove that
\[
\hat \rho_n =  \sign (\overline Y_m) = \sign \left (\frac{Y_{1,m} + Y_{2,m} + \cdots + Y_{m,m}}{m} \right ), \quad k_m \leq n <  k_{m+1}
\]
is a consistent estimator for the ancestral state.
Without the loss of generality, we assume that $\rho = \sign(\mu) = 1$.
We have
\begin{align*}
E(\overline Y_m) &= \frac{1}{m} \sum_{i=1}^m{\left [P(Z_{i,m} > 0) - P(Z_{i,m} < 0) \right ]} = \left (\frac{2}{m} \sum_{i=1}^m{P(Z_{i,m} > 0)} \right ) - 1 \\
& = \left (\frac{2}{m} \sum_{i=1}^m{\Phi \left(\frac{\mu}{\sigma\sqrt{t_{i,m}}}\right )} \right ) - 1 \geq 2 \Phi \left(\frac{\mu}{\sigma\sqrt{H}}\right ) -1 > 0
\end{align*}
where $\Phi$ is the cumulative distribution function of the standard Normal distribution.
Hence,
\begin{align*}
P(\hat \rho_n = -1) &= P( \overline Y_m < 0) = P[\overline Y_m - E(\overline Y_m) < - E(\overline Y_m)] \\
&\leq P[| \overline Y_m - E(\overline Y_m) | \geq E(\overline Y_m)] \leq \frac{\var(\overline Y_m)}{E(\overline Y_m)^2} \\
& \leq \frac{\var(\overline Y_m)}{\left [2 \Phi \left(\frac{\mu}{\sigma\sqrt{H}}\right ) -1 \right ]^2}.
\end{align*}
Note that $\var(Y_{i,m}) \leq 1$ since $|Y_{i,m}| = 1$.
By Lemma \ref{lem:boundcorr}, we have
\[
|\cov(Y_{i,m}, Y_{j,m})| \leq \frac{t_{ij,m}}{\sqrt{t_{i,m} t_{j,m}}} \sqrt{\var(Y_{i,m}) \var(Y_{j,m})} \leq \frac{t_{ij,m}}{\sqrt{t_{i,m} t_{j,m}}} \leq \frac{1}{m \alpha}.
\]
Therefore,
\begin{align*}
\var(\overline Y_m) &= \left (\sum_{i=1}^m{\var(Y_{i,m})} + 2\sum_{1 \leq i < j \leq m}{\cov(Y_{i,m}, Y_{j,m})} \right ) \\
&\leq \frac{1}{m^2} \left (m + m(m-1)\frac{1}{m \alpha} \right ) \leq \frac{1+\alpha^{-1}}{m} \to 0.
\end{align*}
We conclude that 
\[
P(\hat \rho_n = 1) \leq \frac{\var(\overline Y_m)}{\left [2 \Phi \left(\frac{\mu}{\sigma\sqrt{H}}\right ) -1 \right ]^2} \to 1.
\] 
Thus, $\hat \rho_n $ is a consistent estimator.

\end{proof}

\begin{remark}
We complete the proof of Theorem \ref{thm:main} by combining Lemmas \ref{lem:Ane} and \ref{lem:Roch} with Theorems \ref{thm:BMnecessary}, \ref{thm:BMbigbang}, \ref{thm:thresholdneccesary}, and  \ref{thm:thresholdsufficient}.
\end{remark}

\section{Unbounded heights}

When the tree heights are unbounded, the equivalence between the condition $\one\transpose\V_n^{-1}\one \to + \infty$ and the big bang condition is no longer valid.
To see this, let us consider a simple scenario where $(\tree_n)_{n=1}^\infty$ is a sequence of nested star tree (see Figure \ref{fig:star}).
Let $H_n$ be the distance from the root to the $n$-th species.
It is trivial that the sequence of trees $(\tree_n)_{n=1}^\infty$ satisfies the big bang condition.
On the other hand, 
\[
\one\transpose\V_n^{-1}\one = \sigma^2 \sum_{n=1}^\infty{\frac{1}{H_n}}.
\]
Hence, $\one\transpose\V_n^{-1}\one \to + \infty$ if and only if 
\[
\sum_{n=1}^\infty{\frac{1}{H_n}} \to + \infty.
\]
Therefore, the condition $\one\transpose\V_n^{-1}\one \to + \infty$ and the big bang condition are not equivalent when tree heights are unbounded.

\begin{figure}[h]
\centering
\includegraphics[width = 0.6\textwidth]{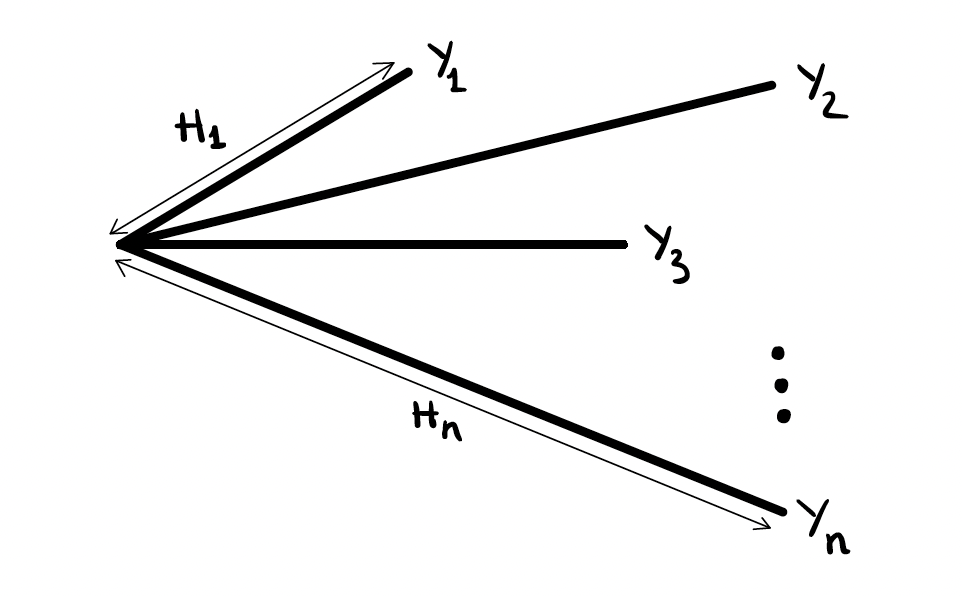}
\caption{An $n$-species star tree.}
\label{fig:star}
\end{figure}

We note that Theorem \ref{thm:BMnecessary} and Theorem \ref{thm:thresholdneccesary} do not require the heights of trees are bounded. 
Therefore, even without the bounded heights condition, $\one\transpose\V_n^{-1}\one \to + \infty$ is the necessary condition for the existence of a consistent ancestral state reconstruction method under the BM model and threshold model.
On the other hand, the big bang condition is the necessary condition for the existence of a consistent estimator under regular discrete models when tree heights are not bounded \citep[see the proof of Proposition 3.1 in][]{fan2018necessary}.
A natural question is when tree heights are unbounded, whether either the big bang condition or the condition $\one\transpose\V_n^{-1}\one \to + \infty$ is a sufficient condition for these models.
\citet{ane2008analysis} gives a positive answer for the BM model by showing that the MLE is consistent if $\one\transpose\V_n^{-1}\one \to + \infty$.
Unfortunately, without additional conditions, neither condition is enough to guarantee that there is a consistent method for reconstructing the ancestral state under regular discrete models.
Specifically, we provide a simple counter-example using a sequence of nested star trees and the two-state symmetric model.

\begin{thm}
Consider a sequence of nested star trees $(\tree_n)_{n=1}^\infty$.
Let $H_n$ be the height of $\tree_n$ such that $H_n / n \to 0$ and $H_n / \log n \to \infty$.
Then, 
\[
\sum_{n=1}^\infty{\frac{1}{H_n}} \to + \infty,
\]
but there is no consistent ancestral state reconstruction method under the two-state symmetric model.
\label{thm:counter-ex}
\end{thm}

\begin{proof}
Let $P_{\rho}$ be the joint distribution of the observations $\mb{Y}_n = (Y_1, Y_2, \ldots, Y_n)$ under the two-state symmetric model with ancestral state $\rho$.
Denote $p_k = [1 + \exp(- \eta H_k)]/2$ where $\eta$ is the mutation rate of the binary trait.
We have 
\begin{align*}
\KL(P_{1}, P_{0}) &= \sum_{k=1}^n {p_k \log \left (\frac{p_k}{1-p_k} \right ) + (1-p_k) \log \left (\frac{1 - p_k}{p_k} \right )} \\
& = \sum_{k=1}^n {(2 p_k -1 ) \log \left (1 + \frac{2p_k - 1}{1-p_k} \right )} \\
& \leq \sum_{k=1}^n { \frac{(2p_k - 1)^2}{1-p_k}} \leq C \sum_{k=1}^n {(2p_k - 1)^2} \\
&= C \sum_{k=1}^n {\exp(- \eta H_k)} = C \sum_{k=1}^n {\left (\frac{1}{k}\right )^{\eta H_k/\log k}} < + \infty.
\end{align*}
By the same arguments of Theorems \ref{thm:BMnecessary} and \ref{thm:thresholdneccesary}, we deduce that there is no consistent estimator for the ancestral state.
\end{proof}

\section{Conclusion and Discussion}

In this work, we provide a unified theory for ancestral state reconstruction across different models for a sequence of nested trees with bounded heights.
We show that the condition $\one\transpose\V_n^{-1}\one \to + \infty$ arose from the study of the BM model is equivalent to the big bang condition for discrete models. 
Furthermore, these conditions are the necessary and sufficient condition for the existence of a consistent estimator for the ancestral state under the BM, regular discrete, and threshold models.

We provide a simple counter-example to show that when tree heights are unbounded, the condition $\one\transpose\V_n^{-1}\one \to + \infty$ and the big bang condition are no longer equivalent. 
Moreover, neither condition is a sufficient condition for the existence of a consistent ancestral state reconstruction method under regular discrete models.
It is worth noticing that the condition $\one\transpose\V_n^{-1}\one \to + \infty$ is the necessary and sufficient condition for the existence of a consistent estimator for the ancestral state under the BM model without the requirement of bounded tree heights.
Establishing a necessary and sufficient condition for regular discrete models and the threshold model when tree heights are unbounded remains open.

It is worth noticing that the MLE for the ancestral state is consistent under the BM model when the condition $\one\transpose\V_n^{-1}\one \to + \infty$ holds \citep{ane2008analysis}.
Furthermore, by Proposition 6 in \citet{steel2008maximum}, when evolution dynamics of a regular finite-state discrete model is known, the MLE for the ancestral state is consistent if there exists a consistent ancestral state reconstruction method.
However, the consistency of the MLE under the threshold model remains unknown.
In some scenarios, there exists a better ancestral state reconstruction method than the MLE \citep{ho2019multi, ho2021ancestral}.
Therefore, the condition $\one\transpose\V_n^{-1}\one \to + \infty$, which also implies the big bang condition, may not be a sufficient condition for the consistency of the MLE under the threshold model.

\section*{Acknowledgement}
LSTH was supported by startup funds from Dalhousie University, the Canada Research Chairs program, the
NSERC Discovery Grant RGPIN-2018-05447, and the NSERC Discovery Launch Supplement DGECR-2018-00181.
VD was supported by a startup fund from the University of Delaware and National Science Foundation grant DMS-1951474.

\bibliographystyle{chicago}
\bibliography{biblio}

\begin{thebibliography}{}

\bibitem[\protect\citeauthoryear{An{\'e}}{An{\'e}}{2008}]{ane2008analysis}
An{\'e}, C. (2008).
\newblock Analysis of comparative data with hierarchical autocorrelation.
\newblock {\em Annals of Applied Statistics\/}~{\em 2\/}(3), 1078--1102.

\bibitem[\protect\citeauthoryear{An{\'e}, Ho, and Roch}{An{\'e}
  et~al.}{2017}]{ane2017phase}
An{\'e}, C., L.~S.~T. Ho, and S.~Roch (2017).
\newblock Phase transition on the convergence rate of parameter estimation
  under an {Ornstein--Uhlenbeck} diffusion on a tree.
\newblock {\em Journal of Mathematical Biology\/}~{\em 74\/}(1), 355--385.

\bibitem[\protect\citeauthoryear{Fan and Roch}{Fan and
  Roch}{2018}]{fan2018necessary}
Fan, W.-T.~L. and S.~Roch (2018).
\newblock Necessary and sufficient conditions for consistent root
  reconstruction in markov models on trees.
\newblock {\em Electronic Journal of Probability\/}~{\em 23}.

\bibitem[\protect\citeauthoryear{Faria, Rambaut, Suchard, Baele, Bedford, Ward,
  Tatem, Sousa, Arinaminpathy, P{\'e}pin, et~al.}{Faria
  et~al.}{2014}]{faria2014early}
Faria, N.~R., A.~Rambaut, M.~A. Suchard, G.~Baele, T.~Bedford, M.~J. Ward,
  A.~J. Tatem, J.~D. Sousa, N.~Arinaminpathy, J.~P{\'e}pin, et~al. (2014).
\newblock The early spread and epidemic ignition of {HIV-1} in human
  populations.
\newblock {\em Science\/}~{\em 346\/}(6205), 56--61.

\bibitem[\protect\citeauthoryear{Felsenstein}{Felsenstein}{1985}]{felsenstein1985phylogenies}
Felsenstein, J. (1985).
\newblock Phylogenies and the comparative method.
\newblock {\em The American Naturalist\/}~{\em 125\/}(1), 1--15.

\bibitem[\protect\citeauthoryear{Felsenstein}{Felsenstein}{2004}]{felsenstein2004inferring}
Felsenstein, J. (2004).
\newblock {\em Inferring phylogenies}, Volume~2.
\newblock Sinauer associates Sunderland, MA.

\bibitem[\protect\citeauthoryear{Felsenstein}{Felsenstein}{2012}]{felsenstein2012comparative}
Felsenstein, J. (2012).
\newblock A comparative method for both discrete and continuous characters
  using the threshold model.
\newblock {\em The American Naturalist\/}~{\em 179\/}(2), 145--156.

\bibitem[\protect\citeauthoryear{Gill, Ho, Baele, Lemey, and Suchard}{Gill
  et~al.}{2017}]{gill2017relaxed}
Gill, M.~S., L.~S.~T. Ho, G.~Baele, P.~Lemey, and M.~A. Suchard (2017).
\newblock A relaxed directional random walk model for phylogenetic trait
  evolution.
\newblock {\em Systematic Biology\/}~{\em 66\/}(3), 299--319.

\bibitem[\protect\citeauthoryear{Ho and An{\'e}}{Ho and
  An{\'e}}{2013}]{ho2013asymptotic}
Ho, L. S.~T. and C.~An{\'e} (2013).
\newblock Asymptotic theory with hierarchical autocorrelation:
  {Ornstein--Uhlenbeck} tree models.
\newblock {\em The Annals of Statistics\/}~{\em 41\/}(2), 957--981.

\bibitem[\protect\citeauthoryear{Ho and An{\'e}}{Ho and
  An{\'e}}{2014}]{ho2014intrinsic}
Ho, L. S.~T. and C.~An{\'e} (2014).
\newblock Intrinsic inference difficulties for trait evolution with
  {Ornstein--Uhlenbeck} models.
\newblock {\em Methods in Ecology and Evolution\/}~{\em 5\/}(11), 1133--1146.

\bibitem[\protect\citeauthoryear{Ho, Dinh, and Nguyen}{Ho
  et~al.}{2019}]{ho2019multi}
Ho, L. S.~T., V.~Dinh, and C.~V. Nguyen (2019).
\newblock Multi-task learning improves ancestral state reconstruction.
\newblock {\em Theoretical Population Biology\/}~{\em 126}, 33--39.

\bibitem[\protect\citeauthoryear{Ho and Susko}{Ho and
  Susko}{2021}]{ho2021ancestral}
Ho, L. S.~T. and E.~Susko (2021).
\newblock Ancestral state reconstruction with large numbers of sequences and
  edge-length estimation.
\newblock {\em arXiv preprint arXiv:2104.00151\/}.

\bibitem[\protect\citeauthoryear{Lancaster}{Lancaster}{1957}]{lancaster1957some}
Lancaster, H.~O. (1957).
\newblock Some properties of the bivariate normal distribution considered in
  the form of a contingency table.
\newblock {\em Biometrika\/}~{\em 44\/}(1/2), 289--292.

\bibitem[\protect\citeauthoryear{Li, Steel, and Zhang}{Li
  et~al.}{2008}]{li2008more}
Li, G., M.~Steel, and L.~Zhang (2008).
\newblock More taxa are not necessarily better for the reconstruction of
  ancestral character states.
\newblock {\em Systematic Biology\/}~{\em 57\/}(4), 647--653.

\bibitem[\protect\citeauthoryear{Liberles}{Liberles}{2007}]{liberles2007ancestral}
Liberles, D.~A. (2007).
\newblock {\em Ancestral sequence reconstruction}.
\newblock Oxford University Press on Demand.

\bibitem[\protect\citeauthoryear{Maddison}{Maddison}{1994}]{maddison1994phylogenetic}
Maddison, D.~R. (1994).
\newblock Phylogenetic methods for inferring the evolutionary history and
  processes of change in discretely valued characters.
\newblock {\em Annual Review of Entomology\/}~{\em 39\/}(1), 267--292.

\bibitem[\protect\citeauthoryear{Revell}{Revell}{2014}]{revell2014ancestral}
Revell, L.~J. (2014).
\newblock Ancestral character estimation under the threshold model from
  quantitative genetics.
\newblock {\em Evolution\/}~{\em 68\/}(3), 743--759.

\bibitem[\protect\citeauthoryear{Steel}{Steel}{2016}]{steel2016phylogeny}
Steel, M. (2016).
\newblock {\em Phylogeny: discrete and random processes in evolution}.
\newblock SIAM.

\bibitem[\protect\citeauthoryear{Steel and Rodrigo}{Steel and
  Rodrigo}{2008}]{steel2008maximum}
Steel, M. and A.~Rodrigo (2008).
\newblock Maximum likelihood supertrees.
\newblock {\em Systematic Biology\/}~{\em 57\/}(2), 243--250.

\bibitem[\protect\citeauthoryear{Thornton}{Thornton}{2004}]{thornton2004resurrecting}
Thornton, J.~W. (2004).
\newblock Resurrecting ancient genes: experimental analysis of extinct
  molecules.
\newblock {\em Nature Reviews Genetics\/}~{\em 5\/}(5), 366--375.

\bibitem[\protect\citeauthoryear{Vajda}{Vajda}{1970}]{vajda1970note}
Vajda, I. (1970).
\newblock Note on discrimination information and variation (corresp.).
\newblock {\em IEEE Transactions on Information Theory\/}~{\em 16\/}(6),
  771--773.

\end{thebibliography}

\end{document}